\documentclass[aps,prx,superscriptaddress,twocolumn,twoside,nofootinbib,floatfix,a4paper]{revtex4-2} %
\pdfoutput=1
\usepackage[T1]{fontenc}
\usepackage[utf8]{inputenc}
\usepackage{amsmath}
\usepackage{mathtools}
\usepackage{microtype}
\usepackage{amssymb}
\usepackage{amsthm}
\usepackage{mathrsfs}
\usepackage{bm, dsfont}
\usepackage[usenames,dvipsnames]{color}
\usepackage{colortbl}
\usepackage[table]{xcolor}
\usepackage{enumitem}

\usepackage{graphicx}
\usepackage[colorlinks=true,linkcolor=blue,citecolor=blue,urlcolor=blue]{hyperref}
\usepackage{natbib}
\usepackage{hypcap}
\usepackage{verbatim, float}
\usepackage{psfrag}
\usepackage{tikz}
\usetikzlibrary{arrows.meta, automata,positioning,quotes}
\usepackage[normalem]{ulem}
\usepackage{physics}

\newtheorem{theorem}{Theorem}
\newtheorem*{theorem*}{Theorem}
\newtheorem{result}{Result}

\newtheorem*{corollary}{Corollary}
\newtheorem{observation}{Observation}

\newtheorem*{lemma*}{Lemma}
\newtheorem*{result*}{Result}

\usepackage{array} 
\usepackage{multirow} 

\usepackage{diagbox} 

\newcommand{\id}{\mathds{1}}
\newcommand{\LC}{\mathtt{LC}}

\begin{document}

\title{All pure entangled states can lead to fully nonlocal correlations}


\author{Martin J. Renner}
    \thanks{These two authors contributed equally.}
   \affiliation{ICFO - Institut de Ciencies Fotoniques, The Barcelona Institute of Science and Technology, 08860 Castelldefels, Spain}
\author{Edwin Peter Lobo }
     \thanks{These two authors contributed equally.}
    \affiliation{Laboratoire d'Information Quantique, Université libre de Bruxelles (ULB), Belgium}
\author{Arturo Konderak}
   \affiliation{Center for Quantum-Enabled Computing, Center for Theoretical Physics, Polish Academy of Sciences, al. Lotników 32/46, 02-668 Warsaw, Poland}
\author{Remigiusz Augusiak}
    \affiliation{Center for Quantum-Enabled Computing, Center for Theoretical Physics, Polish Academy of Sciences, al. Lotników 32/46, 02-668 Warsaw, Poland}
\author{Antonio Acín}
    \affiliation{ICFO - Institut de Ciencies Fotoniques, The Barcelona Institute of Science and Technology, 08860 Castelldefels, Spain}
    \affiliation{ICREA - Instituci{\'o} Catalana de Recerca i Estudis Avan\c cats, Llu\'\i s Companys 23, 08010 Barcelona, Spain}

\date{\today}

\begin{abstract}
    It is a well-established fact that some quantum correlations can be nonlocal, meaning that they cannot be described by a local hidden variable model. Certain quantum correlations have a form of nonlocality so strong that they cannot be reproduced even by models having an arbitrarily small local hidden variable component. These correlations are called fully nonlocal and lead to Bell inequalities in which the maximum quantum value saturates the non-signaling bound.
    A well-known example of this effect, which is also referred to as quantum pseudo-telepathy or all-versus-nothing proofs of nonlocality, is the quantum distribution fulfilling the Peres-Mermin square, in which the underlying state is a $4\times4$ dimensional maximally entangled state. Other examples of full nonlocality are known but, so far, all of them  are for maximally entangled states and it is an open question whether maximal entanglement is necessary for full nonlocality. In this work, we first establish a link between full nonlocality and the concept of antidistinguishability of quantum states. We use this connection to show that in every bipartite $d\times d$ Hilbert space, with $d\geq3$, there are non-maximally entangled states that are fully nonlocal. In fact, we derive simple sufficient conditions for full nonlocality that are only based on the smallest and largest Schmidt coefficients. We also show that in every dimension there exist pure entangled states that do not exhibit full nonlocality. Finally, we show that all pure entangled states can be activated to show full nonlocality in the many-copy scenario.
\end{abstract}

\maketitle

\section{Introduction}
Bell nonlocality is a striking feature of quantum theory and challenges our classical understanding of nature~\cite{Bell, Brunner_2014}. In a typical Bell experiment, the parties perform local measurements on a shared entangled quantum state. If the state and measurements are suitably chosen, then the resulting statistics cannot be explained by a local hidden variable (LHV) model. Typically, this is certified by violating a Bell inequality, where the quantum strategy outperforms every strategy that admits an LHV model. However, in most Bell experiments, the quantum bound of the Bell inequality is smaller than the maximum value that can be attained by general non-signaling theories. A well-known example is the Clauser-Horne-Shimony-Holt (CHSH) inequality where the Tsirelson bound is smaller than the non-signaling value that can be attained by Popescu-Rohrlich correlations~\cite{CHSH1969, Cirelson1980, popescu_quantum_1994}. This naturally raises the question of whether there exist Bell inequalities for which a quantum strategy can attain the largest possible value compatible with the no-signaling principle.

Famous examples of such Bell inequalities are the Mermin inequality~\cite{mermin_extreme_1990} for an odd number of parties, maximally violated by the  Greenberger, Horne, and Zeilinger (GHZ) state~\cite{Greenberger1989, ghz1990}, and the Peres-Mermin (or magic) square construction for the bipartite case~\cite{peres_incompatible_1990,mermin_simple_1990,cabello_all_2001,aravind_quantum_2004}. Both Bell inequalities can be cast as games that can be won with certainty by a quantum strategy. In contrast, any classical strategy fails with nonzero probability. Since the winning probability cannot exceed one, the quantum bound is equal to the non-signaling bound for these Bell inequalities. 

This effect, which we call full nonlocality in this work and that is also known as (or equivalent to) quantum pseudo-telepathy or all-versus-nothing proofs of nonlocality~\cite{Liu2024}, is arguably the strongest form of nonlocality. Beyond its foundational importance, it serves as a resource for shallow quantum circuits~\cite{Bravyi2018, Bravyi2020,bharti_power_2023}, robust randomness certification~\cite{coudron_trading_2021}, and device-independent quantum key distribution (DI-QKD)~\cite{jain_parallel_2020,zhen2023}, often outperforming standard nonlocal strategies. Given these significant advantages, determining which states exhibit full nonlocality is a pivotal objective, both for theoretical understanding and for practical implementations.

Although the notion of full nonlocality was already introduced decades ago~\cite{stairs_quantum_1983, heywood_nonlocality_1983,ghz1990,mermin_extreme_1990,mermin_simple_1990,peres_incompatible_1990,mermin_1993, Gisin2006}, determining which quantum states exhibit full nonlocality remains an open question. In the bipartite case, it is known that no pure non-maximally entangled two-qubit state exhibits full nonlocality~\cite{Elitzur1992,Scarani2008,Branciard2010,Portmann2012,Renner2023}. On the other hand, all maximally entangled states exhibit full nonlocality: via infinitely many measurements for local dimension $d=2$~\cite{Elitzur1992,barrett2006},
and finite Kochen--Specker sets for $d\geq 3$~\cite{RennerWolf2004,brassard_minimum_2005,aravind_quantum_2004,cabello_all_2001, cabello_simplest_2025, cabello_simplest_2025_2}.
Apart from maximally entangled states and qubit pairs, it is generally unknown whether other states exhibit fully nonlocal correlations in the bipartite setting~\cite{Mancinska2015}.

In this work, we first show that full nonlocality extends to a broad class of pure non-maximally entangled states,
see Theorem~\ref{maintheorem}. The core insight behind our results is a connection between full nonlocality and the antidistinguishability of quantum states, a notion introduced by Caves, Fuchs, and Schack~\cite{Caves2002} and later used by Pusey, Barrett, and Rudolph to challenge the epistemic view of quantum states~\cite{PBR2012}. Beyond its foundational significance, antidistinguishability has been used to demonstrate quantum advantages in communication~\cite{Perry2015, Heinosaari2019, Havlicek2020, Bae2026}, and, in recent years, general theorems have been established on the antidistinguishability of quantum states~\cite{Jain2014, Heinosaari2018, Russo2023, Johnston2025}.

Building on this connection between full nonlocality and antidistinguishability, we prove that non-maximally entangled states that are fully nonlocal exist in every $d\times d$ space, with $d\geq 3$. We further show that full nonlocality can be activated for every pure entangled state in the many-copy scenario: more precisely, for any given pure entangled state, there exists a number of copies such that the resulting system generates fully nonlocal correlations. Finally, we prove that, in every dimension, there exist pure states that do not exhibit full nonlocality at the single-copy level, even when allowing for general (not necessarily projective) measurements.

\section{Preliminaries}

\begin{figure}
    \centering
    \includegraphics[width=0.9\linewidth]{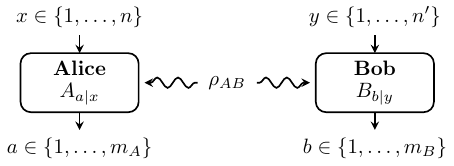}
    \caption{In a Bell scenario, two parties obtain the outputs $a, b$ based on the inputs $x,y$, generating the joint probability distributions $p(a,b|x,y)$. In a quantum realization, these probability distributions are derived from the measurements on a shared entangled state via the Born rule~\eqref{def:quantum_corr}.}
    \label{fig:bell_scenario}
\end{figure}

\subsection{Local content}\label{sec:local_content}
In the standard Bell scenario (see Fig.~\ref{fig:bell_scenario}), two spatially separated parties, Alice and Bob, choose inputs $x \in  [n] \coloneqq \{1,\dots,n\} $ and $y \in [n'] $ for their respective devices and obtain outputs $a \in [m_A]$ and $b \in [m_B]$, with probability given by the distribution $p(a,b|x,y)$.
We assume that the distribution is non-signaling, i.e., Alice's marginal distribution is independent of Bob's input
\begin{align}\label{eq:nonsignalling}
   \forall y,y', \,\, p_A(a|x)=\sum_{b}p(a,b|x,y)=\sum_{b}p(a,b|x,y')\,,
\end{align}
and an analogous condition holds for Bob. Any bipartite non-signaling distribution can be decomposed as 
\begin{equation}\label{eq:local_content_decomp}
    p(a,b|x,y)=q\cdot p_L(a,b|x,y)+(1-q)\cdot p_{NL}(a,b|x,y)\,,
\end{equation}
where $q\in[0,1]$, $p_L(a,b|x,y)$ is a local distribution that can be written as
\begin{align}\label{eq:localpart}
    p_L(a,b|x,y)=\sum_\lambda p(\lambda)\cdot  p_A(a|x,\lambda)\cdot p_B(b|y,\lambda) \, ,
\end{align}
and $p_{NL}(a,b|x,y)$ is an arbitrary probability distribution.\footnote{$p_{NL}$ is also non-signaling since it is the difference of two non-signaling distributions $p$ and $p_L$.} The local content ($\LC$) of a distribution $p$ is defined as the maximum value of $q$ in any such decomposition. If $\LC = 1$, then clearly $p$ is a local distribution. If $\LC <1$, then $p$ is said to be nonlocal. 
Of special interest to us are the distributions $p$ for which $\LC = 0$. We call such correlations that have no local content fully nonlocal, and we give a simple condition below to identify them.

Recall that the local distribution $p_L$ in Eq.~\eqref{eq:local_content_decomp} can be further decomposed into local deterministic strategies, i.e.,
\begin{equation}\label{eq:decomp_local}
    p_L(a,b|x,y)= \sum_{\boldsymbol{\alpha},\boldsymbol{\beta}} p(\boldsymbol{\alpha},\boldsymbol{\beta}) \, D_{\boldsymbol{\alpha},\boldsymbol{\beta}}(a,b|x,y)\,,
\end{equation}
where the tuple $(\boldsymbol{\alpha} = [\alpha_1,\dots,\alpha_n],\boldsymbol{\beta}=[\beta_1,\dots,\beta_{n'}])$ corresponds to the deterministic local strategy in which Alice (Bob) outputs $a=\alpha_x$ ($b={\beta_y}$) for the input $x\in[n]$ ($y\in[n']$) so that $D_{\boldsymbol{\alpha},\boldsymbol{\beta}}(a,b|x,y) \coloneqq \delta_{{\alpha}_x,a}\, \delta_{{\beta}_y,b}$. 

A simple observation that is crucial in what follows is that, if $p(a,b|x,y)=0$ for some $(a,b,x,y)$, then a local deterministic distribution $D_{\boldsymbol{\alpha},\boldsymbol{\beta}}$ that appears in the decomposition of $p$ as in Eq.~\eqref{eq:local_content_decomp} (with weight $q \cdot p(\boldsymbol{\alpha},\boldsymbol{\beta}) >0$ in Eq.~\eqref{eq:decomp_local}) must have $D_{\boldsymbol{\alpha},\boldsymbol{\beta}}(a,b|x,y)=0$, i.e., 
\begin{equation}\label{eq:zeroQ_gives_zeroL}
D_{\boldsymbol{\alpha},\boldsymbol{\beta}}(a,b|x,y) = 0\quad \forall (a,b,x,y) \text{ s.t. } p(a,b|x,y)=0.
\end{equation}
Therefore, if a distribution $D_{\boldsymbol{\alpha},\boldsymbol{\beta}}$ violates the above condition, i.e., if for some $(a,b,x,y)$ we have $p(a,b|x,y) = 0$ and $D_{\boldsymbol{\alpha},\boldsymbol{\beta}}(a,b|x,y) = 1$, then it cannot appear in the decomposition of $p$. If \textit{every} local deterministic distribution violates the above condition, then $p$ is fully nonlocal.
See Table~\ref{tab:simpleexample} for an illustration of this method for the Popescu-Rohrlich box, which is the simplest fully nonlocal (but not quantum) distribution. Further, we show in  App.~\ref{appsec:fnl_vs_antidistinguishability} that the converse also holds, i.e., if a distribution $p$ is fully nonlocal, then every local deterministic distribution $D_{\boldsymbol{\alpha},\boldsymbol{\beta}}$ violates condition~\eqref{eq:zeroQ_gives_zeroL} for some $(a,b,x,y)$ (see Eqs.~\eqref{appeq:nonzero_local_content},~\eqref{appeq:nonzero_local_content2}).
\begin{table}
\newcommand{\tallrow}{\rule[-1.1ex]{0pt}{3.6ex}}
    \centering
   \begin{tabular}{cccc}
        $p_{\text{PR}}=$ & \begin{tabular}{c | c || cc | cc } 
    & $y$ & \multicolumn{2}{c|}{1} & \multicolumn{2}{c}{2} \\ \cline{1-6}
    $x$ & {\diagbox[width=1.4em, height=1.6em]{$a$}{$b$}} & $1$ & $2$ &  $1$ & $2$   \\  \hline \hline
    
    \multirow{2}{*}{1} & 1 & \tallrow {$\frac{1}{2}$} & 0 & $\frac{1}{2}$ & {\underline{\color{magenta}0}} \\  
                       & 2 & \tallrow 0 & $\frac{1}{2}$ & 0 & $\frac{1}{2}$ \\ \cline{1-6} 
    
    \multirow{2}{*}{2} & 1 & \tallrow {$\frac{1}{2}$} & 0 & 0 & {$\frac{1}{2}$} \\ 
                       & 2 & \tallrow 0 & $\frac{1}{2}$ & $\frac{1}{2}$ & 0 \\ 
\end{tabular} & $ \ D_{\boldsymbol{\alpha},\boldsymbol{\beta}}$ = &  
\begin{tabular}{c|  c || cc | cc } 
    & $y$ & \multicolumn{2}{c|}{1} & \multicolumn{2}{c}{2} \\ \cline{1-6}
    $x$ & {\diagbox[width=1.4em, height=1.6em]{$a$}{$b$}} & $1$ & $2$ &  $1$ & $2$   \\  \hline \hline
    
    \multirow{2}{*}{1} & 1 & \tallrow 1 & 0 & 0 & {\underline{\color{magenta}1}} \\  
                       & 2 & \tallrow 0 & 0 & 0 & 0 \\ \cline{1-6} 
    
    \multirow{2}{*}{2} & 1 & \tallrow 1 & 0 & 0 & 1 \\ 
                       & 2 & \tallrow 0 & 0 & 0 & 0 \\ 
\end{tabular}\\
    \end{tabular} \hfill    
    \caption{Left: Probability table for the Popescu-Rohrlich Box $p_{\text{PR}}(a,b|x,y)$. Right: A local deterministic strategy, where Alice outputs $a=1$ for both inputs and Bob outputs $b=y$. The distribution $p_{\text{PR}}$ is fully nonlocal since any of the 16 local deterministic strategies cannot appear in the decomposition of Eq.~\eqref{eq:local_content_decomp}. For the local strategy shown here, the contradiction arises since $p_{\text{PR}}(1,2|1,2)=0$ but $D_{\boldsymbol{\alpha},\boldsymbol{\beta}}(1,2|1,2)=1$.}
    \label{tab:simpleexample}
\end{table}

In this work, we study the local content of quantum correlations $p_Q$ that can be written as
\begin{equation}\label{def:quantum_corr}
    p_Q(a,b|x,y)=\Tr[(A_{a|x}\otimes B_{b|y})\rho_{AB}]\,,
\end{equation}
for some bipartite quantum state $\rho_{AB}$ and local quantum measurements $\{A_x\}_x$ and $\{B_y\}_y$, where $\rho_{AB} \succeq 0$, ${\Tr}[\rho_{AB}]=1$, and $A_x = \{A_{a|x}\}_a$, $B_y = \{B_{b|y}\}_b$ define positive operator-valued measures (POVMs) whose elements are positive $A_{a|x}\succeq 0$ and sum to the identity $\sum_a A_{a|x}=\id$, and similarly for $B_{b|y}$. A quantum state $\rho_{AB}$ is said to be fully nonlocal if there exist local quantum measurements such that the resulting distribution $p_Q$ in Eq.~\eqref{def:quantum_corr} is fully nonlocal. In what follows, we focus on pure states and provide conditions to detect its full nonlocality. While in this work we mostly consider scenarios with finitely many inputs and outputs, we remark that some quantum states, for instance the maximally entangled two-qubit state, can show full nonlocality only if an infinite set of measurements is considered~\cite{barrett2006}.

\subsection{Geometric interpretation of full nonlocality}\label{sec:geometric_int_fnl}

We can think of the probability distribution $p$ as a vector in a real vector space. Distributions satisfying Eqs.~\eqref{eq:nonsignalling},~\eqref{eq:localpart}, and~\eqref{def:quantum_corr} define, respectively, the non‑signaling (NS), local (L), and quantum (Q) sets in this space. A Bell functional $I(p)$ is a linear functional of the probabilities $p$ defined by a set of real coefficients $\{I_{abxy}\}_{a,b,x,y}$, 
\begin{align}\label{eq:bell_functional}
    I(p) \coloneqq \sum_{a,b,x,y} I_{abxy}\ p(a,b|x,y)\,.
\end{align}
We denote by $w_L$ and $w_{NS}$ the largest values of $I(p)$ achievable by local and non-signaling distributions, respectively, and by $w_Q$ the supremum over quantum distributions. 

As proven in Ref.~\cite{Liu2024}, fully nonlocal quantum correlations lie on the boundary of the non-signaling set, and furthermore, they define a perfect quantum strategy for a suitably chosen Bell game (also known as pseudo-telepathy). In fact, it is possible to construct a Bell functional which attains its largest non-signaling value $w_{NS}$ in in the quantum set $Q$, while all local strategies perform worse ($w_L \lneq w_Q=w_{NS}$). While we refer the reader to Ref.~\cite{Liu2024} for further details, we provide a simple example of such a Bell functional. Its coefficients can be defined as:\footnote{It is worth mentioning that applying this procedure to the PR-correlations in Table~\ref{tab:simpleexample} leads to the CHSH inequality (up to rescaling of the coefficients).}
\begin{equation}
    I_{abxy} \coloneqq \begin{cases}
        0 \quad &\text{if} \quad p_Q(a,b|x,y)>0\,,\\
        -1 \quad &\text{if} \quad p_Q(a,b|x,y)=0\,.
    \end{cases}\,
\end{equation}
The Bell functional~\eqref{eq:bell_functional} constructed from these coefficients clearly has the maximal quantum value equal to $0$, and no non-signaling distribution can achieve a larger value, i.e., $w_Q=w_{NS}=0$. At the same time, if $p_Q(a,b|x,y)$ is fully nonlocal, then every local deterministic distribution violates condition~\eqref{eq:zeroQ_gives_zeroL}, and hence cannot avoid all the $-1$s in the Bell functional. Therefore, $w_{L} \leq -1$, which is strictly less than the quantum (or non-signaling) bound. This motivates the picture in Fig.~\ref{fig:sets}.
\begin{figure}
    \centering
    a)\includegraphics[width=0.46\linewidth]{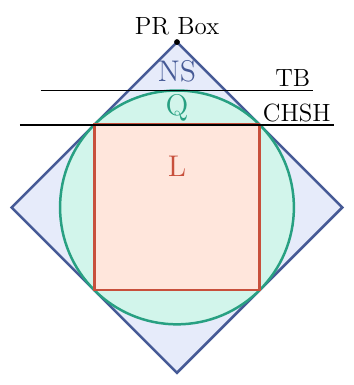}b) \includegraphics[width=0.46\linewidth]{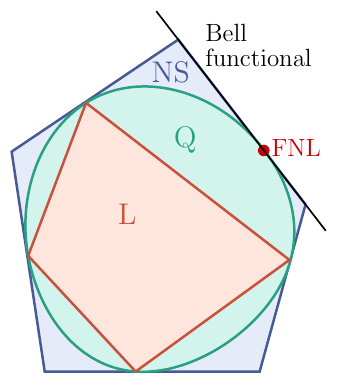}
    \caption{Graphical representations of the local (L), quantum (Q), and non-signaling (NS) set. a) A section of the set of quantum correlations in the two-input two-output scenario. The local polytope (L) is bounded by the CHSH inequality~\cite{goh_geometry_2018}. The Tsirelson's bound (TB) identifies the maximal violation of the CHSH inequality allowed by quantum correlations. The Popescu-Rohrlich (PR) box provides the maximal violation of the CHSH inequality via non-signaling correlations. b) A schematic representation of a fully nonlocal (FNL) quantum distribution. It is a point on the boundary of the non-signaling set that does not contain a vertex of the local polytope.}
    \label{fig:sets}
\end{figure}

On the other hand, if the local content of $p_Q$ is nonzero, then a Bell functional with $w_L < w_Q = w_{NS}$ cannot exist. Indeed, since the Bell functional is linear, a decomposition of the form in Eq.~\eqref{eq:local_content_decomp} implies that $w_Q \leq q\  w_L +(1-q)\ w_{NS}$, which is strictly smaller than $w_{NS}$ unless $w_L=w_{NS}$. Intuitively speaking, one can simulate quantum correlations that have a local content with local resources in a fraction $q$ of the rounds and in these rounds the performance is bounded by the maximal local value.

\subsection{Antidistinguishability of quantum states}
One of the main results of this work is to establish a connection between full nonlocality of quantum correlations and the concept of antidistinguishability. A set of states $\{\rho_i\}_{i \in [n]}$ is called antidistinguishable if there exists a measurement $\{M_i\}_{i \in [n]}$ such that  $\mathrm{Tr}[M_i\rho_i]=0$ for all $i \in [n]$ \cite{Caves2002, Leifer2014, Heinosaari2018}. Hence, 
the outcome $i$ implies that the prepared state was not $\rho_i$. This is also called state exclusion and it can be solved via the following semidefinite program (SDP)~\cite{Jain2014, Johnston2025}:
\begin{equation}\label{eq:SDP}
\begin{aligned}
    \min_{\{M_{i}\}_i}&\ \sum_i \mathrm{Tr}[M_i  \rho_{i}] \\ \text{subject to}&:\ \sum_i M_i=\id\,,\, M_i \succeq 0  .
\end{aligned}
\end{equation}
Since $\mathrm{Tr}[M_i \rho_i]\geq 0$, the set $\{\rho_i\}_{i\in[n]}$ is antidistinguishable if and only if the solution to the SDP above is $0$.
We say that a measurement $\{M_j\}_{j\in[k]}$ is an antidistinguishing measurement for $\{\rho_j\}_{j\in[n]}$ if and only if it satisfies\footnote{Note that there are slightly different notions of antidistinguishability in the literature. The one that we have considered is sometimes also called weak antidistinguishability~\cite{Skrikumar2024}.}
\begin{equation}\label{eq:antidistinguishing_measurement}
        \forall M \in \{M_j\}_{j\in[k]} \,, \exists \rho \in \{\rho_j\}_{j\in[n]} \, \text{s.t.} \, \mathrm{Tr}[M\rho]=0\,.
\end{equation}
In App.~\ref{app:anti_meas}, we show that the existence of such a measurement is equivalent to the antidistinguishability of $\{\rho_j\}_{j\in[n]}$. 
A simple example of states that are antidistinguishable is given by
\begin{equation}\label{eq:antidistinguishability_example}
    \begin{aligned}
    \ket{ \psi_0}&=\alpha_1 \ket{1}+\alpha_2 \ket{2}\,,\\
    \ket{ \psi_1}&=\beta_1 \ket{0}+\beta_2 \ket{2}\,,\\
    \ket{ \psi_2}&=\gamma_1 \ket{0}+\gamma_2 \ket{1} \,,
\end{aligned}
\end{equation}
where $\alpha_i, \beta_i, \gamma_i\in \mathbb{C}$ are arbitrary complex numbers (up to normalization). Clearly, a measurement in the computational basis $M_0=\ketbra{0}$, $M_1=\ketbra{1}$, $M_2=\ketbra{2}$ antidistinguishes these three states. Further, it was shown by Caves, Fuchs, and Schack~\cite{Caves2002} that any three pure states $\{\ket{ \psi_0}, \ket{ \psi_1}, \ket{ \psi_2}\}$ that satisfy $|\braket*{ \psi_i}{ \psi_j}|\leq 1/2$ $(i\neq j)$ can be written as in Eq.~\eqref{eq:antidistinguishability_example} in some basis, and hence are antidistinguishable. Recently, this criterion was generalized by Johnston, Russo, and Sikora in~\cite{Johnston2025}, where the following theorem was established:
\begin{theorem}[{\cite[Cor.~5.3]{Johnston2025}}]\label{thmJohnston}
    A set of $n$ pure quantum states $\{\ket{ \psi_1}, \ket{ \psi_2}, \dots, \ket{ \psi_n}\}$ is antidistinguishable if
    \begin{align}\label{eq:johnston_condition}
        \norm{G}_F\coloneqq\sqrt{\sum_{i=1}^n \sum_{j=1}^n \lvert\braket{ \psi_i}{ \psi_j}\rvert^2}\leq \frac{n}{\sqrt{2}} \, .
    \end{align}
    Here, $G_{ij}=\braket*{ \psi_i}{ \psi_j}$ is the Gram matrix associated with the states $\{\ket{ \psi_i}\}_{i\in[n]}$, and $\norm{G}_F$ is its Frobenius norm. In particular, the set is antidistinguishable if $\lvert\braket{ \psi_i}{ \psi_j}\rvert\leq \sqrt{({n-2})/({2n-2})}$ for all pairs $i\neq j$.
\end{theorem}

\section{Connecting full nonlocality and antidistinguishability}\label{sec:fnl_vs_antidistinguishability}
Consider a Bell test defined by local measurements on a bipartite state $\rho_{AB}$, see Eq.~\eqref{def:quantum_corr}. Denote by $\rho^B_{a|x}$ the post-measurement quantum state on Bob's side when Alice performs the measurement $x\in [n]$ and gets the outcome $a\in\{1,\dots,m_A\}=[m_A]$, i.e.,
\begin{align}
    \rho^B_{a|x}=\frac{\mathrm{Tr_A}[(A_{a|x}\otimes \id)\rho_{AB}]}{\mathrm{Tr}[(A_{a|x}\otimes \id)\rho_{AB}]} \, .
\end{align}
Let us denote by $\boldsymbol{\alpha} = [\alpha_1,\dots,\alpha_n]$, where $\alpha_x \in [m_A]$, some collection of outcomes for every input $x\in[n]$. 
Suppose now that for some fixed $\boldsymbol{\alpha}$ the post-measurement states $\rho^B_{\boldsymbol{\alpha}} \coloneqq \{\rho^B_{a=\alpha_1|x=1}, \dots ,\rho^B_{a=\alpha_n|x=n}\}$ are antidistinguishable. Further suppose that Bob, by choosing a measurement labeled as $y=\boldsymbol{\alpha}$, performs the antidistinguishing measurement $\{B_{b| y=\boldsymbol{\alpha} }\}_{b\in [n]}$. Then, it follows from the definition of antidistinguishability that $\mathrm{Tr}[\rho^B_{a=\alpha_x|x} B_{b=x|y=\boldsymbol{\alpha}}] = 0$ for all values of $x\in[n]$. Therefore, conditioned on Alice giving the outcome $a=\alpha_x$ for $x$, the probability that Bob gives the outcome $b=x$ for $y=\boldsymbol{\alpha}$ is $0$, i.e.,
\begin{equation}
    p_Q(b=x|x,y=\boldsymbol{\alpha},a=\alpha_x) = \mathrm{Tr}[\rho^B_{a=\alpha_x|x} B_{b=x|y=\boldsymbol{\alpha}}]=0\,,
\end{equation}
which further implies, using Bayes' rule $p_Q(a,b|x,y)=p_Q(b|x,y,a)p_Q(a|x,y)$, that
\begin{equation}
    \forall x \in [n]\,, \quad p_Q(a=\alpha_x,b=x|x,y=\boldsymbol{\alpha})=0\,.
\end{equation}
It is now easy to see that any deterministic local strategy $D_{\boldsymbol{\alpha},\boldsymbol{\beta}}$ in which Alice's device outputs the outcomes according to the strategy given by $\boldsymbol{\alpha}$ cannot appear in the decomposition of $p_Q$ as in Eq.~\eqref{eq:local_content_decomp}. Indeed, for any such local deterministic strategy, Bob must assign some outcome for the measurement choice $y=\boldsymbol{\alpha}$, and hence $D_{\boldsymbol{\alpha},\boldsymbol{\beta}}(a=\alpha_x,b=x|x,y=\boldsymbol{\alpha})=1$ for some $x\in[n]$, which violates condition~\eqref{eq:zeroQ_gives_zeroL}.

This directly establishes the following link: consider a quantum distribution $p_Q$~\eqref{def:quantum_corr} such that (i) for all choices of local strategies $\boldsymbol{\alpha}$, the set of post-measurement states $\rho^B_{\boldsymbol{\alpha}}$ are antidistinguishable and (ii) Bob performs the corresponding antidistinguishing measurement for each $\boldsymbol{\alpha}$.  Then, no local deterministic strategy can appear in the decomposition of $p_Q$, and the resulting distribution is fully nonlocal. 

A converse statement also holds: If for a given state and some measurements for Alice there is at least one $\boldsymbol{\alpha}$ such that the post-measurement states $\rho^B_{\boldsymbol{\alpha}}$ are \textit{not} antidistinguishable, then the resulting correlations always have a local content, irrespective of the measurements performed by Bob (see App.~\ref{appsec:fnl_vs_antidistinguishability} for proof). This leads us to the following theorem:

\begin{theorem}\label{thm:antidistinguishable}
    A bipartite state is fully nonlocal if and only if there exist measurements for Alice such that all the sets $\rho^B_{\boldsymbol{\alpha}} \coloneqq \{\rho^B_{a=\alpha_1|x=1},\dots,\rho^B_{a=\alpha_n|x=n}\}$ of post-measurement states are antidistinguishable.
\end{theorem}
Note that the number of measurement settings $n'$ required for Bob can in principle be quite large, equal to $(m_A)^n$. But this number can decrease (sometimes significantly) if the same measurement can antidistinguish several sets of post-measurement states $\rho^B_{\boldsymbol{\alpha}}$ as it happens, for instance, in the Peres-Mermin square. In this nonlocal game~\cite{aravind_quantum_2004}, Alice and Bob each perform $3$ measurements with $4$ outcomes each. There are thus $4^3=64$ different local deterministic strategies for Alice which correspond to an equal number of sets of post-measurement states which must be antidistinguishable. However, only $3$ measurements on Bob's side are sufficient to antidistinguish all of them. See App.~\ref{app:Peresmermin} for a detailed description. 

In what follows, we make use of the sufficient condition in Theorem~\ref{thmJohnston} for antidistinguishability of quantum states, together with the aforementioned theorem, to find general classes of states that are fully nonlocal. 

\section{Activation of pure entangled qubit pairs}\label{sec:activating_qubits}
As a simple first example, we show that all pure entangled qubit pairs are fully nonlocal in the multi-copy scenario using a finite number of measurements. In contrast, it is known that with a single copy only the maximally entangled two-qubit state is fully nonlocal and further the corresponding distribution requires an infinite number of measurement settings \cite{Elitzur1992,barrett2006, Scarani2008, Branciard2010}.

Consider first a general two-qubit state (we assume $p\geq 1/2$ without loss of generality)
\begin{equation}
    \ket{ \psi}_{AB}=\sqrt{p}\ket{00}+\sqrt{1-p}\ket{11} \, ,
\end{equation}
and suppose Alice has three measurement settings corresponding to the qubit observables $X$, $Y$, and $Z$. The post-measurement states on Bob's side conditioned on Alice obtaining the outcome $a\in\{0,1\}$ for the measurement choice $x \in \{X,Y,Z\}$ are given by $\ket*{ \psi^B_{a|Z}}=\ket{a}$, $\ket*{ \psi^B_{a|X}}=\sqrt{p}\ket{0}+(-1)^a\sqrt{1-p}\ket{1}$, and $\ket*{ \psi^B_{a|Y}}=\sqrt{p}\ket{0}+(-1)^a\,\mathrm i \sqrt{1-p}\ket{1}$. Calculating the inner products of these post-measurement states and using $\sqrt{p^2+(1-p)^2}\leq \sqrt{p}$ for $p\geq 1/2$, we obtain, for $x,x'\in\{X,Y,Z\}$ and $a,a'\in\{0,1\}$,
\begin{equation}
    |\braket*{ \psi^B_{a|x}}{ \psi^B_{a'|x'}}|\leq \sqrt{p}\quad (x \neq x')\,.
\end{equation}
Suppose now that Alice and Bob share $k$ copies of this two-qubit state $\ket{ \psi}_{AB}$ and once again Alice has three measurement settings corresponding to performing one of $X$, $Y$, or $Z$ basis measurement on all the individual qubits. That is, Alice measures along the same basis on all of her $k$ qubits, and obtains the $k$-tuple outcome $\vec{a} = (a_1,\dots,a_k) \in \{0,1\}^k$. The post-measurement states on Bob's side are then the tensor products of single copy post-measurement states, i.e.,
\begin{equation}
    \begin{aligned}
    &\ket*{ \psi^{B}_{\vec{a}|Z}} = \ket{\vec{a}} = \ket{a_1\dots a_k}\,,\\
    &\ket*{ \psi^{B}_{\vec{a}|X}} =\bigotimes_{i=1}^k \left( \sqrt{p}\ket{0}+(-1)^{a_i} \sqrt{1-p}\ket{1}\right)\, ,\\
    &\ket*{ \psi^{B}_{\vec{a}|Y}} = \bigotimes_{i=1}^k \left(\sqrt{p}\ket{0}+\mathrm i(-1)^{a_i} \sqrt{1-p}\ket{1}\right)\, ,
    \end{aligned}
\end{equation}
and satisfy, for all $x,x'\in\{X,Y,Z\}$, 
\begin{equation}
    |\braket*{ \psi^B_{\vec{a}|x}}{ \psi^B_{\vec{a}'|x'}}|\leq \sqrt{p}^k\quad (x \neq x')\,,
\end{equation}
due to the tensor product structure.
If $\sqrt{p}^k\leq 1/2$, then all the sets of post-measurement states $\{\ket*{ \psi^{B}_{\vec{a}|X}}$, $\ket*{ \psi^{B}_{\vec{a}'|Y}}$, $\ket*{ \psi^{B}_{\vec{a}''|Z}}\}$ are antidistinguishable (Theorem~\ref{thmJohnston}) and hence the state $(\sqrt{p}\ket{00}+\sqrt{1-p}\ket{11})^{\otimes k}$ demonstrates full nonlocality (Theorem~\ref{thm:antidistinguishable}).

For the maximally entangled state, $p=1/2$, two copies, $k=2$, are enough to demonstrate full nonlocality with a finite number of measurement settings, but the full nonlocality of this state was already known as it is equivalent to a maximally entangled state in dimension $4\times 4$. Beyond maximal entanglement, our results show that any two-qubit entangled state can be activated if sufficiently many copies are available, since there always exists a value of $k$ such that $\sqrt{p}^k \leq 1/2$ for $p<1$. It is worth noting that we have used only three separable measurements for Alice. Later on, we use similar ideas to extend this result to arbitrary $d$-dimensional pure entangled states.

\section{Full nonlocality from Mutually Unbiased Bases}\label{sec:sufficient_condition_fnl_mub}
In this section, we derive some simple sufficient criteria for full nonlocality. Consider a general pure state written in its Schmidt form
\begin{equation}\label{eq:Schmidt_form}
    \ket{ \psi}_{AB}=\sum_{i=0}^{d-1} \sqrt{\lambda_i} \ket{ii} \, .
\end{equation}
Suppose Alice measures this state along one of $n$ mutually unbiased bases (MUBs), i.e., the POVM elements (projectors) of Alice's $n$ measurements $\{A_x\}_{x\in[n]}$ with $d$ outcomes each are of the form $A_{a|x}=\ketbra*{A_{a|x}}$ and satisfy, for $x,x'\in[n]$ and $a,a' \in \{0,\dots,d-1\}$,
\begin{equation}\label{eq:defMUB}
    |\braket*{A_{a|x}}{A_{a'|x'}}|=\frac{1}{\sqrt{d}} \quad (x\neq x')\,.
\end{equation}
It is known that at least three MUBs exist in every dimension $d$, while the exact number for a given dimension is not known~\cite{mcnulty_mutually_2024}. If Alice and Bob share a maximally entangled state ($\lambda_i=1/d$ for all $i$), the post-measurement state on Bob's side conditioned on Alice obtaining the outcome $a$ for the measurement $x$ is given by $\ket*{\psi^B_{a|x}} = \ket*{A_{a|x}}^*$, and thus 
\begin{align}
    |\braket*{ \psi^B_{a|x}}{ \psi^B_{a'|x'}}|=\frac{1}{\sqrt{d}} \quad (x\neq x')\,.
\end{align}
Using Theorems~\ref{thmJohnston} and~\ref{thm:antidistinguishable}, this directly proves that all maximally entangled states in dimension $d\geq 4$ exhibit full nonlocality, which also follows from Ref.~\cite{RennerWolf2004} by using the connection to Kochen-Specker sets. If we instead have a non-maximally entangled state that is close to being maximally entangled (i.e., all the Schmidt coefficients $\lambda_i$ are close to $1/d$), then we intuitively expect that the overlaps $|\braket*{ \psi^B_{a|x}}{ \psi^B_{a'|x'}}|$ above are close to $1/\sqrt{d}$. As long as these overlaps are below a given bound (e.g., $1/2$ for three measurements), then we have full nonlocality due to Theorems~\ref{thmJohnston} and~\ref{thm:antidistinguishable}.

More precisely, let us fix Alice's first measurement to be along the Schmidt basis, i.e., $\ket*{A_{a|x=1}}=\ket*{a}$. Due to Eq.~\eqref{eq:defMUB}, the other measurements are then of the form
\begin{equation}
    \ket{A_{a|x}}=\sum_{j=0}^{d-1} \braket*{j}{A_{a|x}} \ket*{j}= \frac{1}{\sqrt{d}}\,\sum_{j=0}^{d-1} \mathrm e^{\mathrm i\phi^j_{a|x}} \ket{j} \quad (x \neq 1),
\end{equation}
where we have used $\braket*{j}{A_{a|x}}=\braket*{A_{j|x=1}}{A_{a|x}} = e^{\mathrm i\phi^j_{a|x}}/\sqrt{d}$ for some phase $e^{\mathrm i\phi^j_{a|x}}$ whose precise value is not relevant for the discussion. If Alice measures along these MUBs, the post-measurement states on Bob's side conditioned on Alice obtaining the outcome $a$ for the measurement choice $x$ are given by
\begin{equation}\label{eq:postmeasurement_states_MUB}
    \ket*{ \psi^B_{a|x}}=
    \begin{cases}
        \ket*{a} \quad &\text{if}\quad x=1\,,\\
        \sum_{j=0}^{d-1} \sqrt{\lambda_j}\ \mathrm e^{-\mathrm i\phi^j_{a|x}} \ket{j}  \quad &\text{if} \quad x \neq 1\, .
    \end{cases}
\end{equation} 
As shown in App.~\ref{app:overlap_bounds}, we can bound the overlaps of these post-measurement states as follows
\begin{align}
    &\forall  x'>x\,,\nonumber\\
        &|\braket*{ \psi^B_{a|x}}{ \psi^B_{a'|x'}}|\leq
        \begin{cases}
        \sqrt{\lambda_{max}}  &\text{if}\quad x=1\,, \\ 
               1-(d-\sqrt{d})\lambda_{min}  &\text{if} \quad x \neq 1\,,
    \end{cases}\label{eq:upperbound1}
\end{align}
with $\lambda_{max}$ and $\lambda_{min}$ the largest and smallest Schmidt coefficients of $\ket{\psi}_{AB}$. This allows us to apply Theorems~\ref{thmJohnston} and~\ref{thm:antidistinguishable} to arrive at the following theorem.
\begin{theorem}\label{maintheorem}
    Suppose there are at least $n$ MUBs in a Hilbert space of dimension $d$. Then, a state $\ket{ \psi}_{AB}=\sum_{i=0}^{d-1} \sqrt{\lambda_i} \ket{ii}$ is fully nonlocal if
\begin{equation}\label{complicatedcorrolary}
     \frac{2}{n}\lambda_{max}+\frac{n-2}{n}\left[1-(d-\sqrt{d})\lambda_{min}\right]^2\leq \frac{n-2}{2n-2} \,.
\end{equation}
In particular, this is true if
\begin{equation}\label{simplecorrolary}
     \max\left(\sqrt{\lambda_{max}},\, 1-(d-\sqrt{d})\lambda_{min}\right) \leq \sqrt{ \frac{n-2}{2n-2} }  \, .
\end{equation}
\end{theorem}
We refer the reader to App.~\ref{app:main_thm} where a stronger version of this theorem is proved.
\begin{result}\label{res:result_1}
Non-maximally entangled states with full nonlocality exist in all dimensions $d\geq 3$.
\end{result}
\begin{proof}
For systems of dimension $d=3$, also known as qutrits, it is known that the number of MUBs is $n=4$, and the bounds in Theorem~\ref{maintheorem} are satisfied if and only if $\lambda_{max}=\lambda_{min}=1/3$. This demonstrates that the maximally entangled state is fully nonlocal, but tells us nothing about partially entangled states. However, in App.~\ref{appsec:qutrits_fnl} we provide another sufficient condition for full nonlocality using $5$ qutrit measurements for Alice, thus proving that fully nonlocal non-maximally entangled states exist in $d=3$.
For $d=4$ there are exactly $5$ MUBs and the condition in Eq.~\eqref{simplecorrolary} for full nonlocality becomes
\begin{equation}
    \max\left( \sqrt{\lambda_{max}}, \, 1-2\lambda_{min}\right)\leq \sqrt{ \frac{3}{8} } \, ,
\end{equation}
which is satisfied whenever $\lambda_{max}\leq 3/8=0.375$ and $\lambda_{min}\geq (1/2)-\sqrt{ 3/32 }\approx 0.194$. 

For $d\geq 5$, we note that there are at least three MUBs~\cite{mcnulty_mutually_2024}, and the condition in Eq.~\eqref{simplecorrolary} can be simplified to
\begin{equation}
    \max\left( \sqrt{\lambda_{max}}, \, 1-(d-\sqrt{d})\lambda_{min}\right)\leq \frac{1}{2}\,,
\end{equation}
which is satisfied whenever $\lambda_{max}\leq 1/4$ and $\lambda_{min}\geq {1}/(2(d-\sqrt{d}))$. It is simple to check that there are nontrivial values of $\lambda_{min}$ and $\lambda_{max}$ within these bounds such that $\lambda_{min} < \lambda_{max}$ whenever $d \ge 5$. Hence, there are partially entangled states that demonstrate full nonlocality also in every dimension $d\geq 5$.
\end{proof}
\begin{corollary}
    There are states with a relatively large coefficient $\lambda_{max}> 1/2$ that demonstrate full nonlocality. 
\end{corollary}
These states are interesting because they cannot be deterministically transformed to any maximally entangled state by local operations and classical communication (LOCC) \cite{Nielsen1999}. In fact, while not following from the proofs, one could intuitively expect that the previous results on full nonlocality of non-maximally entangled states could be understood in terms of the full nonlocality of the maximally entangled state, when allowing for arbitrary preprocessing via LOCC. The previous corollary rules out this possibility. To construct such an example, choose a large prime power dimension, say $d=101$, where it is known that $n=d+1=102$ MUBs exist. Eq.~\eqref{complicatedcorrolary} then shows that the state $\ket{ \psi}_{AB}=\sqrt{0.67}\ket{00}+\sum_{j=1}^{101} \sqrt{0.0033} \ket{jj}$ is fully nonlocal. The smallest dimension in which we can construct such an example is $d=8$, for which $9$ MUBs exist and the state $\ket{ \psi}_{AB}=\sqrt{0.51}\ket{00}+\sum_{j=1}^{7} \sqrt{0.07} \ket{jj}$ is fully nonlocal according to Eq.~\eqref{complicatedcorrolary}.
\begin{result} 
    For every pure entangled bipartite state $\ket*{\psi}_{AB}$, there exists an integer $k$ such that the $k-$copy state $\ket*{\psi}_{AB}^{\otimes k}$ is fully nonlocal.
\end{result}
\begin{proof}
This activation of full nonlocality is a simple extension of the qubit case in Sec.~\ref{sec:activating_qubits}. It follows from Eq.~\eqref{eq:upperbound1} that when Alice and Bob share a single copy pure entangled state of the form Eq.~\eqref{eq:Schmidt_form} and Alice measures along the $n\geq 3$ MUBs, the overlaps of the post-measurement states obey, for $x\neq x'$,
\begin{equation}
    |\braket*{ \psi^B_{a|x}}{ \psi^B_{a'|x'}}|\leq \max\left( \sqrt{\lambda_{max}},\,  1-(d-\sqrt{d})\lambda_{min}\right)\,.
\end{equation}
The right hand side in the equation above may not be sufficient to certify full nonlocality of a single copy of the entangled state. However, for every pure entangled state ($\lambda_{max}< 1$ and $\lambda_{min}> 0$) we have $c\coloneqq\max (\sqrt{\lambda_{max}}, 1-(d-\sqrt{d})\lambda_{min})<1$ whenever $d \geq 2$. Choose an integer $k$ that is large enough to satisfy $c^k\leq \sqrt{({n-2})/({2n-2})}$. If Alice and Bob share $k$ copies of the entangled state and Alice performs one of the $n$ MUB measurements on each copy like in the case for qubit pairs, the inner product of the resulting post-measurement states is upper bounded by $c^k$. Theorem~\ref{thmJohnston} then guarantees that all the resulting sets of post-measurement states $\ket*{ \psi_{\boldsymbol{\alpha}}}$ that appear in Theorem~\ref{thm:antidistinguishable} are antidistinguishable, and hence the distribution is fully nonlocal. 
\end{proof}
Note that the strategy above to generate fully nonlocal correlations is not optimal in the number of copies required.  For instance, consider two copies of the two qubit state $\ket{ \psi}_{AB}^{\otimes 2}=(\sqrt{p}\ket{00}+\sqrt{1-p}\ket{11})^{\otimes2}$ with $p\ge 1/2$, which can also be seen as a single copy state in dimension $d=4$ with $\lambda_{max}=p^2$ and $\lambda_{min}=(1-p)^2$. Using Eq.~\eqref{simplecorrolary}, this state is fully nonlocal if $\lambda_{max}=p^2\leq 3/8$ and $\lambda_{min}=(1-p)^2\geq 1/2-\sqrt{3/32}$ which is true when $1/2\leq p\leq 1-\sqrt{1/2-\sqrt{3/32}}\approx 0.560$. Hence, some non-maximally entangled two-qubit states only need two copies to show full nonlocality. However, if Alice performs separable qubit MUBs on the two copies as in the strategy described above (there are exactly three MUBs for $d=2$ which are precisely the $X$, $Y$, and $Z$ qubit observables used in Sec.~\ref{sec:activating_qubits}), only the two-qubit maximally entangled state can be shown to be fully nonlocal.

The measurements that we have used to activate full nonlocality do not depend on the specific entangled state. We might find better bounds and more examples when adapting the measurements according to the state. We have also assumed very little about the structure of the MUBs. It might be possible to derive better bounds by using the explicit form of the MUBs (i.e., the explicit phases $\phi^j_{a|x}$), or increasing the number of bases (even if it means allowing larger scalar products). In fact, this is precisely what allowed us to find fully nonlocal partially entangled two-qutrit states in App.~\ref{appsec:qutrits_fnl}.

\section{Local content for some partially entangled states}
So far, we asked which states exhibit full nonlocality; now we want to ask which states do not. Recall that the local content $\LC$ of a quantum state $\rho_{AB}$ is defined as the smallest local content of any distribution that can be generated by performing local quantum measurements on it~\cite{Elitzur1992,Scarani2008}.  Consider again an arbitrary finite dimensional pure entangled state as in Eq.~\eqref{eq:Schmidt_form}. We show that if the maximal eigenvalue $\lambda_{max}$ is large enough, then the state must have a nonzero local content. For simplicity, we will assume $\lambda_0\geq\lambda_1\geq\cdots\geq\lambda_{d-1}$ without loss of generality, so that $\lambda_{max} = \lambda_0$.

Recall that if there exists a set of outcomes $\boldsymbol{\alpha}=[\alpha_x]_x$ for every input $x$ of Alice and a set of outcomes $\boldsymbol{\beta}=[\beta_y]_y$ for every input $y$ of Bob such that 
\begin{equation}
    \forall x,y, \quad p_Q(a=\alpha_x,b=\beta_y|x,y)\ge q,
\end{equation}
then the distribution $p_Q$ has local content at least $q$ (see Eqs.~\eqref{appeq:nonzero_local_content},~\eqref{appeq:nonzero_local_content2}). Therefore, to show that a state $\ket{\psi}_{AB}$ has a nonzero local content, it is enough to prove that for every extremal measurement $A_x = \{A_{a|x}\}_a$ of Alice there is always an outcome $a$, and for every extremal measurement $B_y = \{B_{b|y}\}_b$ of Bob there is always an outcome $b$, such that
\begin{equation}
    \Tr[A_{a|x} \otimes B_{b|y}\  \ketbra{\psi}{\psi}_{AB}] =  p_Q(a,b|x,y) \geq q> 0 \,.
\end{equation}
We will use the following theorem, which is proved in App.~\ref{app:localcontent}, to lower bound the local content of partially entangled states.
\begin{theorem}\label{thm:local_content}
    Consider a $d$-dimensional entangled state $\ket{ \psi}=\sum_{i=0}^{d-1}\sqrt{\lambda_i}\ket{ii}$, where $\lambda_0\geq\lambda_1\geq\cdots\geq\lambda_{d-1}$. If Alice and Bob perform rank-$1$ measurements with at most $m_A$ and $m_B$ outcomes respectively on this state, then, for every measurement choice of Alice and Bob, there exists a POVM element $A = \ketbra{A}$ for Alice and $B=\ketbra{B}$ for Bob such that
    \begin{equation}\label{eq:local_content_bound}
     \lvert \braket{ \psi}{A\otimes B} \rvert \ge \frac{\sqrt{\lambda_0}-\sqrt{\lambda_1}\sqrt{(m_A-1)(m_B-1)}}{\sqrt{m_A m_B}}\,.
    \end{equation}
\end{theorem}
The above inequality is nontrivial when $\lambda_0 > \lambda_1 (m_A-1)(m_B-1)$ and demonstrates a nonzero local content for the correlations generated since $\lvert \braket{ \psi}{A\otimes B} \rvert^2 ={\Tr}[\ketbra{ \psi}A \otimes B] > 0$. We now use this theorem to prove our results. We first consider the case when Alice and Bob perform projective measurements, where we can obtain the following lower bound on the local content.
\begin{result}\label{localcontentproj}
    A $d$-dimensional entangled state $\ket{ \psi}=\sum_{i=0}^{d-1}\sqrt{\lambda_i}\ket{ii}$, where $\lambda_0\geq\lambda_1\geq\cdots\geq\lambda_{d-1}$ and $\lambda_0 > \lambda_1 (d-1)^2$, has a nonzero local content $\LC$ at least 
    \begin{equation}\label{eq:local_content_bound_projective}
        \LC \ge \left[\frac{\sqrt{\lambda_0}-\sqrt{\lambda_1}(d-1)}{d}\right]^2 \, ,
    \end{equation}
    if Alice and Bob perform projective measurements.
\end{result}
\begin{proof}
    The proof is a straightforward application of Theorem~\ref{thm:local_content} to projective measurements. We consider first rank-$1$ projective measurements. Since there are exactly $d$ outcomes for a rank-$1$ projective measurement, we can substitute $m_A = m_B = d$ in the inequality~\eqref{eq:local_content_bound} to obtain
    \begin{equation}
        \begin{aligned}
                {\Tr}[A \otimes B \ketbra{ \psi} ] &= \lvert \braket{ \psi}{A\otimes B} \rvert^2 \\
                &\ge \left[\frac{\sqrt{\lambda_0}-\sqrt{\lambda_1}(d-1)}{d}\right]^2 ,
        \end{aligned}
    \end{equation} 
    which must hold for some projector $A$ of Alice and $B$ of Bob for every measurement choice.
    Furthermore, we can always split a rank-$k$ projector into $k$ rank-$1$ projectors, increasing the number of outcomes associated with the projector from $1$ to $k$. The new distribution generated by such splitting of projectors can be used to obtain the original distribution by classically binning the $k$ outcomes into a single outcome. Such local classical post-processing cannot decrease the local content, and hence the new distribution (obtained by splitting the projectors) has a local content that is less than or equal to that of the original distribution. Therefore, proving a lower bound on the local content for rank-$1$ projective measurements is equivalent to proving a lower bound for all projective measurements. 
\end{proof}
For $d=2$, this directly shows that all partially entangled states (i.e., $\lambda_0>\lambda_1$) have a nonzero local content with projective measurements, which was known before~\cite{Elitzur1992,Scarani2008,Branciard2010,Portmann2012,Renner2023}. Further, for the special class of states in which $\lambda_1 = \cdots = \lambda_{d-1}$, we have $\lambda_0+\lambda_1 (d-1) =1$, and hence the condition $\lambda_0 > \lambda_1 (d-1)^2$ for nonzero local content is equivalent to $\lambda_0 > (d-1)/d$. For $d=3$, this gives us an analytic proof for the result by Scarani \cite{Scarani2008} where it was shown numerically that when $\lambda_0>2/3$, the state $\ket{ \psi} = \sqrt{\lambda_0} \ket{00} + \sqrt{\lambda_1} (\ket{11} + \ket{22})$ has a nonzero local content.

We next remove the restriction to projective measurements and generalize our result to arbitrary POVMs. 
\begin{result}\label{res:local_content_bound_povm}
    A $d$-dimensional entangled state $\ket{ \psi}=\sum_{i=0}^{d-1}\sqrt{\lambda_i}\ket{ii}$, where $\lambda_0\geq\lambda_1\geq\cdots\geq\lambda_{d-1}$ and $\lambda_0 > \lambda_1 (d^2-1)^2$, has a nonzero local content $\LC$ at least
    \begin{equation}\label{eq:local_content_bound_povm}
        \LC \ge \left[\frac{\sqrt{\lambda_0}-\sqrt{\lambda_1}(d^2-1)}{d^2}\right]^2 \, .
    \end{equation}
\end{result}
\begin{proof}
    We can use Naimark's dilation theorem~\cite{holevo_quantum_2012}, which implies that a $d$-dimensional POVM can be expressed as a projective measurement in a  $d^2$-dimensional Hilbert space, and the bound simply follows by replacing $d$ with $d^2$ in Result~\ref{localcontentproj}. Alternatively, one can use the fact that in dimension $d$ every non-extremal POVM can be written as a convex mixture of extremal rank-$1$ POVMs with at most $d^2$ outcomes~\cite{Dariano2005}. The result then follows by putting $m_A=m_B=d^2$ in Theorem~\ref{thm:local_content} to obtain
    \begin{equation}
        \begin{aligned}
                {\Tr}[\ketbra{ \psi} A \otimes B] &= \lvert \braket{ \psi}{A\otimes B} \rvert^2 \\
                &\ge \left[\frac{\sqrt{\lambda_0}-\sqrt{\lambda_1}(d^2-1)}{d^2}\right]^2 ,
        \end{aligned}
    \end{equation} 
    which must hold for some POVM element $A$ of Alice and $B$ of Bob for every measurement choice. 

\end{proof}

\section{Conclusion}
In this work, we established a fundamental link between full nonlocality of bipartite quantum correlations and antidistinguishability of quantum states. Exploiting this connection, we derived simple sufficient criteria that certify when a given pure entangled state is fully nonlocal, and we showed that full nonlocality is not confined to maximally entangled states: non-maximally entangled states exhibiting full nonlocality exist in every local dimension $d\geq 3$. We further proved that for any pure entangled state there exists a finite number of copies for which the resulting system is fully nonlocal. 
At the same time, we identified families of pure states in every dimension whose local content is strictly positive, and which therefore do not exhibit full nonlocality even when arbitrary measurements are allowed in the single-copy scenario. Closing the gap between these sufficient conditions by characterizing exactly which pure states are fully nonlocal is an interesting open question. Extending the analysis to mixed states, generalizing the antidistinguishability-based framework to multipartite correlations~\cite{Greenberger1989, Renner2002,almeida_multipartite_2010,aolita_fully_2012,makuta_all_2025}, and understanding the implications for device-independent tasks are some important directions for future work.


\section{Acknowledgements}
E.P.L. acknowledges hospitality from the QIT group at ICFO.

M.J.R. acknowledges financial support through the Juan de la Cierva postdoctoral fellowship (La ayuda JDC2024-055405-I, financiada por MICIU/AEI/10.13039/501100011033 y por el FSE+.) M.J.R. and A.A. acknowledge financial support through the Government of Spain (Severo Ochoa CEX2019-000910-S and FUNQIP), Fundació Cellex, Fundació Mir-Puig, Generalitat de Catalunya (CERCA program), the European Union (NEQST 101080086) and the AXA Chair in Quantum Information Science.

E.P.L. acknowledges funding from the F.R.S.--FNRS through PDR T.0171.22 and from the FWO and F.R.S.--FNRS under the Excellence of Science (EOS) programme (project 40007526). E.P.L. also acknowledges support from the F.R.S.--FNRS through a FRIA grant. Part of this work was carried out during a research visit of E.P.L. to ICFO, supported by the F.R.S.--FNRS mobility funding.

A.K. acknowledges the support by the (Polish) National Science Center through the SONATA BIS Grant No. 2019/34/E/ST2/00369. The Center for Quantum-Enabled Computing project is carried out within the International Research Agendas programme of the Foundation for Polish Science co-financed by the European Union under the European Funds for Smart Economy 2021-2027 (FENG).

\bibliography{bibfilemartin2.bib}

\onecolumngrid
\appendix

\section{Characterization of antidistinguishing measurements}\label{app:anti_meas}
\begin{observation}\label{appobs:antidistinguishability}
    A set of states $\{\rho_j\}_{j\in[n]}$ is antidistinguishable if and only if there exists a measurement $\{M_j\}_{j\in[k]}$ satisfying
    \begin{equation}\label{appeq:antidistinguishing_measurement}
        \forall M \in \{M_j\}_{j\in[k]} \,, \exists \rho \in \{\rho_j\}_{j\in[n]} \, \textnormal{s.t.} \, \Tr[M\rho]=0\,,
    \end{equation}
    that is, an antidistinguishing measurement.
\end{observation}
\begin{proof}
    First, we show that $\{\rho_j\}_{j\in[n]}$ is antidistinguishable if Eq.~\eqref{appeq:antidistinguishing_measurement} is satisfied. Define $N_i$ as the sum of all the operators $M \in \{M_j\}_{j\in [k]}$ such that $\Tr[M \rho_i]=0$, i.e., $N_i \coloneqq \sum_{M\in\mathcal{S}_i} M$\,, where $\mathcal{S}_i = \{M\in\{M_j\}_{j\in[k]}:\Tr[M\rho_i]=0\}$. If the set $\mathcal{S}_i$ is empty for some $i\in[n]$, then define $N_i\coloneqq0$. Clearly, for all $i\in[n]$, $\Tr[N_i \rho_i]=0$. However, $\{N_i\}_{i\in[n]}$ may not form a valid measurement if some $M \in \{M_j\}_{j\in[k]}$ appears in two or more sets $\mathcal{S}_i$. To rectify the situation, we can pick some set $\mathcal{S}_i$ containing $M$ and delete $M$ from all the other sets $\mathcal{S}_j\neq \mathcal{S}_i$. Thus, we have constructed a valid measurement $\{N_i\}_{i\in[n]}$ that satisfies, for all values of $i\in [n]$, $\Tr[N_i \rho_i]=0$, which shows that the set $\{\rho_j\}_{j\in[n]}$ is antidistinguishable.

    The `only if' statement is trivial, since an antidistinguishable set $\{\rho_j\}_{j\in[n]}$ implies, by definition, that there exists a measurement $\{M_j\}_{j\in[n]}$ satisfying Eq.~\eqref{appeq:antidistinguishing_measurement}, i.e., $\forall M_i \in \{M_j\}_{j\in[n]} \,, \exists \rho_i \in \{\rho_j\}_{j\in[n]} \, \text{s.t.} \, \Tr[M_i\rho_i]=0\,$.
\end{proof}

\section{Proof of Theorem~\ref{thm:antidistinguishable}}\label{appsec:fnl_vs_antidistinguishability}
In this section we prove the `only if' statement of Theorem~\ref{thm:antidistinguishable}: If for a given state and some measurements on Alice's side there is at least one $\boldsymbol{\alpha}$ such that the post-measurement states $\rho^B_{\boldsymbol{\alpha}} \coloneqq \{\rho^B_{a=\alpha_1|x=1},\dots,\rho^B_{a=\alpha_n|x=n}\}$ are not antidistinguishable, then the resulting correlations always have a local content, irrespective of the measurements performed by Bob. 

We begin by recalling that the local content of a quantum distribution $p_Q$ is defined as the maximum value of $q$ in any decomposition of the form
\begin{equation}\label{appeq:local_content_decomp}
    p_Q(a,b|x,y)=  q \cdot \sum_{\boldsymbol{\alpha},\boldsymbol{\beta}} p(\boldsymbol{\alpha},\boldsymbol{\beta}) \, D_{\boldsymbol{\alpha},\boldsymbol{\beta}}(a,b|x,y)+(1-q)\cdot p_{NS}(a,b|x,y)\,,
\end{equation}
where $D_{\boldsymbol{\alpha},\boldsymbol{\beta}}(a,b|x,y) = \delta_{{\alpha}_x,a}\, \delta_{{\beta}_y,b}$. If there exists a value of $(\boldsymbol{\alpha},\boldsymbol{\beta})$ such that 
\begin{equation}\label{appeq:nonzero_local_content}
  \forall x\in[n],y\in[n']\,,\quad p_Q(a=\alpha_x,b=\beta_y|x,y) \ge q > 0 \,,
\end{equation}
then the distribution $p_Q$ has local content at least $q$. Indeed, letting 
\begin{equation}\label{appeq:nonzero_local_content2}
p_{NS}(a,b|x,y) \coloneqq  \frac{p_Q(a,b|x,y)-q \cdot D_{\boldsymbol{\alpha},\boldsymbol{\beta}}(a,b|x,y)}{1-q} \,,
\end{equation}
we can write $p_Q(a,b|x,y)=  q \cdot \, D_{\boldsymbol{\alpha},\boldsymbol{\beta}}(a,b|x,y)+(1-q)\cdot p_{NS}(a,b|x,y)$.
We can now prove the `only if' statement of Theorem~\ref{thm:antidistinguishable}.
Since the states $\rho^B_{\boldsymbol{\alpha}}$ are not antidistinguishable, from Observation~\ref{appobs:antidistinguishability}, there must exist, for every measurement $y$ of Bob, at least one outcome $b$ such that 
\begin{equation}\label{eq:not_antidistinguishable}
  \forall x\in[n]\,,\quad  \mathrm{Tr}[ B_{b|y} \rho^B_{a=\alpha_x|x}] > 0 \,.
\end{equation}
Pick an outcome $b$, for every measurement $y \in [n']$, that satisfies the above equation and define $\boldsymbol{\beta}\coloneqq[\beta_1,\dots,\beta_{n'}]$, where $\beta_y = b$. We then have
\begin{equation}
    p_Q(a=\alpha_x,b=\beta_y|x,y)=\Tr[(A_{a=\alpha_x|x}\otimes B_{b=\beta_y|y})\rho_{AB}] = p_Q(a=\alpha_x|x )\cdot \mathrm{Tr}[B_{b=\beta_y|y}\,\rho^B_{a=\alpha_x|x}] > 0,
\end{equation}
where $p_Q(a=\alpha_x|x )=\Tr[(A_{a=\alpha_x|x}\otimes \id)\rho_{AB}]$. Hence, from the discussion surrounding Eq.~\eqref{appeq:nonzero_local_content}, the deterministic local strategy $D_{\boldsymbol{\alpha},\boldsymbol{\beta}}$ has a nonzero weight in the decomposition of $p_Q$.

\section{The Peres-Mermin square from the perspective of antidistinguishability}\label{app:Peresmermin}
Here we consider a known example of full nonlocality to demonstrate the methods used in this work. In the Peres-Mermin magic square nonlocal game, two noncommunicating players Alice and Bob are asked to fill, with labels either `$+$' or `$-$', a row $x\in\{1,2,3\}$ and column $y\in\{1,2,3\}$, respectively, of a $3\times 3$ grid called the magic square. Alice must always assign an even number of $-$'s for the row whereas Bob must always assign an odd number of $-$'s for the column. Additionally, the labels assigned by Alice and Bob must agree in the intersecting element of the grid. In a local deterministic strategy, Alice and Bob agree beforehand the labels for every row and column. However, no local deterministic strategy can win the magic square game perfectly since the total number of $-$'s in the grid would be even according to Alice and odd according to Bob, leading to a contradiction. Therefore, any non-signaling distribution $p(a,b|x,y)$, where $a\in\{+++,+--,-+-,--+\}$ and $b\in\{---,-++,+-+,++-\}$, that allows Alice and Bob to win the magic square game perfectly is fully nonlocal (see 
Sec.~\ref{sec:geometric_int_fnl}).

There is a quantum distribution $p_Q$ that can perfectly win the magic square game, and is hence fully nonlocal~\cite{aravind_quantum_2004}. In particular, it is conjectured to be the simplest fully nonlocal strategy in dimension $4\times 4$~\cite[Conjecture~1]{cabello_simplest_2025}. The quantum strategy uses the maximally entangled state $\ket{ \psi}=(\ket{00}+\ket{11}+\ket{22}+\ket{33})/2$ and the projective measurements $\{A_x\}_{x}$ ($A_x = \{\ketbra*{A_{a|x}}\}_a$) for Alice and $\{B_y\}_{y}$ ($B_y = \{\ketbra*{B_{b|y}}\}_b$) for Bob shown in Tab.\ref{tab:magic_square}.

\begin{table}[hbt!]
    \centering
    \begin{tabular}{c@{\hspace{0.5em}}c}
    \begin{minipage}[t]{0.55\linewidth}
        \centering
        \[
        \begin{array}{c||c|c|c|}
         & \ket*{A_{a|x=1}} & \ket*{A_{a|x=2}} & \ket*{A_{a|x=3}} \\[6pt]\hline\hline
            & & & \\[-10pt]
        a=+++ &
        \ket{0} &
        \ket{0}+\ket{1}+\ket{2}+\ket{3} &
        \ket{0}-\ket{1}-\ket{2}-\ket{3} \\[6pt]
        a=+-- &
        \ket{2} &
        \ket{0}-\ket{1}+\ket{2}-\ket{3}&
         \ket{0}+\ket{1}-\ket{2}+\ket{3}\\[6pt]
        a=-+- &
        \ket{1} &
        \ket{0}+\ket{1}-\ket{2}-\ket{3} &
        \ket{0}-\ket{1}+\ket{2}+\ket{3} \\[6pt]
        a=--+ &
        \ket{3} &
        \ket{0}-\ket{1}-\ket{2}+\ket{3} &
        \ket{0}+\ket{1}+\ket{2}-\ket{3}
        \end{array}
        \]
        \[
        \begin{array}{c||c|c|c|}
         & \ket*{B_{b|y=1}} & \ket*{B_{b|y=2}} & \ket*{B_{b|y=3}} \\[3pt]\hline\hline
         & & & \\[-10pt]
        b=--- &
        \ket{1}-\ket{3} &
        \ket{2}-\ket{3} &
        \ket{1}-\ket{2} \\[6pt]
        b=-++ &
        \ket{1}+\ket{3} &
        \ket{2}+\ket{3} &
        \ket{1}+\ket{2} \\[6pt]
        b=+-+ &
        \ket{0}-\ket{2} &
        \ket{0}-\ket{1}&
        \ket{0}-\ket{3} \\[6pt]
        b=++- &
        \ket{0}+\ket{2} &
        \ket{0}+\ket{1} &
        \ket{0}+\ket{3}
        \end{array}
        \]
    \end{minipage}
    &
    \begin{minipage}[t]{0.4\linewidth}
               \[
        \begin{array}{cc||*{4}{c}}
            &      & \multicolumn{4}{c}{y=2}\\
            &      & --- & -++ & +-+ & ++-\\\hline\hline
               
                \multirow{4}{*}{$x=1$} & +++  & \textcolor{purple}{0} & \textcolor{purple}{0} & \textcolor{purple}{\times}  & \textcolor{purple}{\times}  \\
                & +--  & \times & \times & \times & \times \\
                & -+-  & \times & \times & \times & \times \\
                & --+  & \times & \times & \times & \times \\\hline

                \multirow{4}{*}{$x=2$}& +++  & \times & \times & \times & \times \\
                & +-- & \times & \times & \times & \times \\
                & -+-  & \textcolor{purple}{\times} & \textcolor{purple}{\times} & \textcolor{purple}{0} & \textcolor{purple}{\times} \\
                & --+  & \times & \times & \times & \times \\\hline

                \multirow{4}{*}{$x=3$} & +++  & \textcolor{purple}{\times} & \textcolor{purple}{\times} & \textcolor{purple}{\times} & \textcolor{purple}{0}\\
                & +-- & \times & \times & \times & \times\\
                & -+- & \times & \times & \times & \times \\
                & --+  & \times & \times & \times & \times
        \end{array}
        \]
    \end{minipage}
    \end{tabular}
        \caption{Left: The unnormalized vectors of the quantum measurements used in the magic square game. Right: A section of the probability table $p_Q$ demonstrating that Alice's local deterministic strategy $\boldsymbol{\alpha} = [+++,-+-,+++]$ (marked in red) is forbidden in the decomposition of $p_Q$. The probabilities marked $\times$ are irrelevant for the argument. Bob's measurement $y=2$ antidistinguishes $\psi^B_{\boldsymbol{\alpha}} \coloneqq \{\ket*{ \psi^B_{a=\alpha_x|x}}\}_x$, and hence every outcome is forbidden for Bob. Indeed, $b=---$ is forbidden because, conditioned on Alice giving the outcome $a=+++$ for $x=1$, the probability of Bob giving the outcome $b=---$ is $0$, i.e., $p_Q(b=---|a=+++,y=2,x=1)=|\braket*{ \psi^B_{a=+++|x=1}}{B_{b=---|y=2}}|^2=0$. Similarly, $b=-++$ is forbidden since it implies that Alice's outcome could not have been $a=+++$ for $x=1$, the outcome $b=+-+$ implies that Alice's outcome could not have been $a=-+-$ for $x=2$, and the outcome $b=++-$ implies that Alice's outcome could not have been $a=+++$ for $x=3$. 
        }\label{tab:magic_square}
\end{table}

It is not difficult to argue that this quantum strategy wins the magic square game.
However, here our aim is to make the connection with antidistinguishability and provide an alternate proof that $p_Q$ is fully nonlocal. 

Observe that the post-measurement states on Bob's side, which are given by $\ket*{ \psi^B_{a|x}} = \ket*{A_{a|x}}$, are such that for every local deterministic strategy $\boldsymbol{\alpha}$ of Alice, the corresponding set of post-measurement states $\psi^B_{\boldsymbol{\alpha}} \coloneqq \{\ket*{ \psi^B_{a=\alpha_x|x}}\}_x$ can be antidistinguished by one of the measurements of Bob. For example, suppose Alice's local deterministic strategy is to output $a=+++$ when $x=1$, $a=-+-$ when $x=2$, and $a=+++$ when $x=3$, i.e., $\boldsymbol{\alpha} = [\alpha_1 = +++, \alpha_2 = -+-, \alpha_3 = +++]$.
Then, the measurement $B_{y=2}$ of Bob is an antidistinguishing measurement for $\psi^B_{\boldsymbol{\alpha}}$. This leads to a probability table for $p_Q$ that forbids the local strategy $\boldsymbol{\alpha}$ from appearing in the decomposition of $p_Q$ as explained in Table~\ref{tab:magic_square}. The proof that $p_Q$ is fully nonlocal then follows from the fact that for every local deterministic strategy $\boldsymbol{\alpha}$ of Alice, the corresponding set $ \psi^B_{\boldsymbol{\alpha}} $ is antidistinguishable, with one of the measurements of Bob being the antidistinguishing measurement.

\newpage
\section{Bounds on the overlap for post-measurement states of MUB measurements}\label{app:overlap_bounds}
Here we prove the bound in Eq.~\eqref{eq:upperbound1}, given here as the following lemma.
\begin{lemma*}
When Alice measures on the state $\ket{ \psi}_{AB}=\sum_{i=0}^{d-1} \sqrt{\lambda_i} \ket{ii}$ along the MUBs given by
\begin{equation}\label{appeq:MUB_measurements}
    \ket{A_{a|x}}=
    \begin{cases}
        \ket{a} \quad &\text{if}\quad x=1\,,\\
     \frac{1}{\sqrt{d}}\,\sum_{j=0}^{d-1} \mathrm e^{\mathrm i\phi^j_{a|x}} \ket{j} &\text{if}\quad x \neq 1,
    \end{cases}
\end{equation}
the post-measurement states on Bob's side 
\begin{equation}\label{appeq:postmeasurement_states_MUB}
    \ket*{ \psi^B_{a|x}}=
    \begin{cases}
        \ket*{a} \quad &\text{if}\quad x=1\,,\\
        \sum_{j=0}^{d-1} \sqrt{\lambda_j}\ \mathrm e^{-\mathrm i\phi^j_{a|x}} \ket{j}  \quad &\text{if} \quad x \neq 1\, ,
    \end{cases}
\end{equation} 
satisfy the following upper bound on their overlap:
\begin{equation}\label{appeq:overlap_bounds}
    \forall  x'>x, \quad
        |\braket*{ \psi^B_{a|x}}{ \psi^B_{a'|x'}}|\leq
        \begin{cases}
        \sqrt{\lambda_{max}}  &\text{if}\quad x=1\,, \\ 
               \min_{{\lambda}\in \{\lambda_j\}_{j}}\left(\sum_j |\lambda_j-{\lambda}| +{\lambda} \sqrt{d}\right)\leq 1-(d-\sqrt{d}\,)\lambda_{min}  &\text{if} \quad x \neq 1\,.
    \end{cases}
\end{equation} 
\end{lemma*}
\begin{proof}
The proof is a straightforward calculation. 
For the first measurement $x=1$ we get
\begin{align}
   \forall \,x'>1\,,\qquad |\braket*{ \psi^B_{a|x=1}}{ \psi^B_{a'|x'}}|=\sqrt{\lambda_{a}}\leq \sqrt{\lambda_{max}} \, .
\end{align}
For the other measurements we find
\begin{equation}\label{eq:scalar_product_MUB}
   \forall \,x'>x\,,x\neq 1\,, \quad \braket*{ \psi^B_{a|x}}{ \psi^B_{a'|x'}} =\sum_j \lambda_j \ \mathrm e^{\mathrm i(\phi^j_{a|x}-\phi^j_{a'|x'})}=\sum_j (\lambda_j-{\lambda}) \ \mathrm e^{\mathrm i(\phi^j_{a|x}-\phi^j_{a'|x'})}+\sum_j {\lambda} \ \mathrm e^{\mathrm i(\phi^j_{a|x}-\phi^j_{a'|x'})}\,,
\end{equation}
where ${\lambda}\in \mathbb{R}$ can be chosen arbitrarily. 
Using the fact that $|\sum_i z_i|\leq \sum_i |z_i|$ for complex numbers $z_i\in \mathbb{C}$ we obtain
\begin{equation}
   \forall \,x'>x\,,x\neq 1\,, \quad |\braket*{ \psi^B_{a|x}}{ \psi^B_{a'|x'}}|
    \leq \sum_j |\lambda_j-{\lambda}| +|{\lambda}| \bigg|\sum_j \mathrm  e^{\mathrm i(\phi^j_{a|x}-\phi^j_{a'|x'})}\bigg|.
\end{equation}
Writing $e^{\mathrm i\phi^j_{a|x}}= \sqrt{d} \braket*{j}{A_{a|x}}$, we get
\begin{equation}
     \forall \,x'>x\,,x\neq 1\,, \quad |\braket*{ \psi^B_{a|x}}{ \psi^B_{a'|x'}}|
    \leq \sum_j |\lambda_j-{\lambda}| +d \, \, |{\lambda}| \, \left|\braket*{A_{a'|x'}}{A_{a|x}}\right| = \sum_j |\lambda_j-{\lambda}| +|{\lambda}| \sqrt{d}\,,
\end{equation}
where we used $|\braket*{A_{a'|x'}}{A_{a|x}}| = 1/\sqrt{d}$ for MUBs in the final step.
Since ${\lambda}$ is arbitrary, we can minimize this expression with respect to ${\lambda}\in \mathbb{R}$ to obtain the best bound. It is not hard to see that we can restrict ${\lambda}$ to be non-negative. We will further show in the next section that the best bound is attained when $\lambda \in \{\lambda_j\}_{j}$. A convenient choice for $\lambda$ (but not always the best one) is to choose the smallest Schmidt coefficient $\lambda=\lambda_{min}$. This leads to the simple expression in Eq.~\eqref{appeq:overlap_bounds}, i.e,
\begin{equation}\label{appeq:simple_overlap_MUB}
    \min_{\lambda \in \{\lambda_j\}_{j}} \left[\sum_j |\lambda_j-{\lambda}| +{\lambda} \sqrt{d}\right]\leq\sum_j |\lambda_j-\lambda_{min}| +\lambda_{min} \sqrt{d}=\sum_j (\lambda_j-\lambda_{min}) +\lambda_{min} \sqrt{d}=1-(d-\sqrt{d})\lambda_{min} \,,
\end{equation}
where we have used the fact that $\lambda_j\geq \lambda_{min}$ and $\sum_j \lambda_j=1$.
\end{proof}

\section{Proof of Theorem~\ref{maintheorem}}\label{app:main_thm}
Here we prove a stronger version of Theorem~\ref{maintheorem} stated below.
\begin{theorem*}
    Suppose there are at least $n$ MUBs in dimension $d$. Then, a state $\ket{ \psi}_{AB}=\sum_{i=0}^{d-1} \sqrt{\lambda_i} \ket{ii}$ is fully nonlocal if
\begin{equation}\label{eq:local_majorization}
     \frac{2}{n}\lambda_{max}+\frac{n-2}{n}\left(\min_{\lambda \in \{\lambda_j\}_{j}} f({\lambda})\right)^2\leq \frac{n-2}{2n-2} \,,
\end{equation}
where, $f({\lambda})=\sum_j |\lambda_j-{\lambda}| +{\lambda} \sqrt{d}$. 
\end{theorem*}
Substituting ${\lambda}=\lambda_{min}$ as in Eq.~\eqref{appeq:simple_overlap_MUB} we recover Eq.~\eqref{complicatedcorrolary} in the main text.
\begin{proof}
Let's suppose that Alice measures along the $n$ MUBs as in Sec.~\ref{sec:sufficient_condition_fnl_mub} so that the post-measurement states are given by Eq.~\eqref{eq:postmeasurement_states_MUB} whose overlaps satisfy Eq.~\eqref{appeq:overlap_bounds}. The main idea is to combine these bounds on the overlaps with Theorem~\ref{thmJohnston} to obtain a sufficient condition for antidistinguishability, which then implies full nonlocality due to Theorem~\ref{thm:antidistinguishable}. Consider the Gram matrix $G$ of the set of post-measurement states $\rho^B_{\boldsymbol{\alpha}} \coloneqq \{ \ket*{\psi^B_{a=\alpha_x|x}}\}_{x\in[n]}$ for any arbitrary $\boldsymbol{\alpha} = [\alpha_x]_{x\in[n]}$. The square of its Frobenius norm $(\norm{G}_F)^2$ is given by
\begin{equation}
    \begin{aligned}
        (\norm{G}_F)^2 &\coloneqq \sum_{x=1}^n \sum_{x'=1}^n |\braket*{ \psi^B_{a=\alpha_x|x}}{ \psi^B_{a=\alpha_{x'}|x'}}|^2 \\
        &= \sum_{x=1}^n |\braket*{ \psi^B_{a=\alpha_x|x}}{ \psi^B_{a=\alpha_{x}|x}}|^2 + 2 \sum_{x=1}^n \sum_{x'>x} |\braket*{ \psi^B_{a=\alpha_x|x}}{ \psi^B_{a=\alpha_{x'}|x'}}|^2\\
        &= n + 2\sum_{x'=2}^n |\braket*{ \psi^B_{a=\alpha_1|1}}{ \psi^B_{a=\alpha_{x'}|x'}}|^2 + 2 \sum_{x=2}^n \sum_{x'>x} |\braket*{ \psi^B_{a=\alpha_x|x}}{ \psi^B_{a=\alpha_{x'}|x'}}|^2\,.
    \end{aligned}
\end{equation}
We can now use Eq.~\eqref{appeq:overlap_bounds} to bound the second and third term to write
\begin{equation}
    (\norm{G}_F)^2 \le n + 2 \lambda_{max} (n-1)+ (n-1)(n-2) \left(\min_{{\lambda}\in \mathbb{R}_+} f({\lambda})\right)^2\,.
\end{equation}
Therefore, from Theorem~\ref{thmJohnston}, all the sets of post-measurement states $\rho^B_{\boldsymbol{\alpha}}$ are antidistinguishable if
\begin{equation}
    n + 2(n-1)\lambda_{max} + (n-1)(n-2) \left(\min_{{\lambda}\in \mathbb{R}_+} f({\lambda})\right)^2\leq \frac{n^2}{2}\,.
\end{equation}
Subtracting $n$ on both sides and then dividing by $n(n-1)$ gives the desired bound. It remains to be shown that the best bound is obtained when $\lambda \in \{\lambda_j\}_{j}$.
Indeed, instead of choosing ${\lambda}=\lambda_{min}$ as in Theorem~\ref{maintheorem}, we can directly minimize the second term above with respect to ${\lambda}$. Suppose $\lambda_i$ is ordered so that $\lambda_{0}\le \lambda_{1}\le...\le \lambda_{d-1}$. The function $f({\lambda})$
is convex and piecewise linear with slope $ f'({\lambda})= \sum_j \text{sgn}(\lambda_j-{\lambda}) +\sqrt{d} = 2k-d+\sqrt{d}$ if $\lambda_k \le {\lambda} \le\lambda_{k+1}$.
Notice that the slope changes only when ${\lambda} \in \{\lambda_j\}_j$. Clearly, the minimal value of $f({\lambda})$ is attained at the point $ {\lambda}$ where the slope changes sign from negative to positive. This happens when ${\lambda}=\lambda_{k}$ for $k =\lceil \frac{d-\sqrt{d}}{2}\rceil$ and gives us the value
\begin{equation}
\begin{aligned}
   \min_{{\lambda}\in \mathbb{R}_+} f({\lambda}) &= f({\lambda_k}) = \sum_j |\lambda_j-\lambda_{k}| +\lambda_{k} \sqrt{d}
    = \sum_{j \le {k}} (\lambda_{k}-\lambda_j) +\sum_{j > k} (\lambda_j-\lambda_{k}) +\lambda_{k} \sqrt{d}\\
    &= k \lambda_k - \sum_{j \le {k}} \lambda_j + \sum_{j > k} \lambda_j - (d-k)\lambda_k + \lambda_{k} \sqrt{d}
    =(2k-d+\sqrt{d}\,)\lambda_{k}+1 - 2\sum_{j \leq k} \lambda_j\,.
\end{aligned}
\end{equation}
This completes the proof.
\end{proof}

\section{Full nonlocality of partially entangled two-qutrit states}
\label{appsec:qutrits_fnl}
We have seen that MUBs are particularly useful for proving full nonlocality. Nevertheless, in dimensions $d=3$ there are $4$ MUBs and our criterion is not sufficient to demonstrate full nonlocality of any partially entangled state. Here, we show that if we instead perfrom the $5$ measurements employed in Ref.~\cite{cabello_simplest_2025_2}, the resulting post-measurement can be shown to be antidistinguishable even for some partially entangled two-qutrit states, thus proving that they exhibit full nonlocality. Alice's $5$ measurement bases are given by the columns of the following unitary matrices
\begin{equation}
    U_1=\mathds 1, U_2=\frac{1}{\sqrt{3}} \begin{bmatrix} 1 & 1 & \omega^2 \\ \omega & 1 & \omega \\ \omega^2 & 1 & 1 \end{bmatrix}, U_3=\frac{1}{\sqrt{3}} \begin{bmatrix} 1 & 1 & \omega^2 \\ \omega & 1 & \omega \\ -\omega^2 & -1 & -1 \end{bmatrix}, U_4=\frac{1}{\sqrt{3}} \begin{bmatrix} 1 & 1 & \omega^2 \\ -\omega & -1 & -\omega \\ \omega^2 & 1 & 1 \end{bmatrix}, U_5=\frac{1}{\sqrt{3}} \begin{bmatrix} -1 & -1 & -\omega^2 \\ \omega & 1 & \omega \\ \omega^2 & 1 & 1 \end{bmatrix}\,,
\end{equation} 
where $\omega =\mathrm e^{\frac{2\pi \mathrm i}{3}}$. Explicitly, Alice's basis vectors are $\ket*{A_{a|x}}=U_x \ket{a}$. Notice that the basis elements $\ket*{A_{a|x}}$ can be written as in Eq.~\eqref{appeq:MUB_measurements} with $d=3$ and $\mathrm e^{\mathrm i\phi^j_{a|x}}/\sqrt{3} = \braket{j}{A_{a|x}}  = \braket{j|U_x}{a}$ $(x\neq1)$. Hence, the post-measurement states on Bob's side are given by Eq.~\eqref{appeq:postmeasurement_states_MUB}. We can now follow the same steps as in App.~\ref{app:overlap_bounds}, and put $\lambda=\lambda_{min}$ in the final step to obtain 
\begin{equation}\label{appeq:overlap_qutrit}
    \forall  x'>x\,,\quad |\braket*{ \psi^B_{a|x}}{ \psi^B_{a'|x'}}|\leq
        \begin{cases}
        \sqrt{\lambda_{max}}  &\text{if}\quad x=1\,, \\ 
               1 - 3\lambda_{min} \left(\,1-|\braket*{A_{a'|x'}}{{A_{a|x}}}| \,\right) &\text{if}\quad  x \neq 1\,.
    \end{cases}
\end{equation}
Further, it is straightforward to compute the overlaps $|\braket*{A_{a'|x'}}{{A_{a|x}}}|$ above. Writing $\braket*{A_{a'|x'}}{{A_{a|x}}} = \bra{a}U_x^\dagger U_{x'}\ket{a'}$, we notice that when $x\neq1$, $x'>x$, any two unitaries $U_x$ and $U_{x'}$ are related by a flip of the sign of either one or two rows. Hence, we can always write $U_{x'} = D U_{x}$, where $D = \id - 2\ketbra{j}$ or $2\ketbra{j}- \id$ for some $j$. Suppose $D = \id - 2\ketbra{j}$. Then,
\begin{equation}
    \forall \,x'>x\,,x\neq 1\,, \quad \braket*{A_{a'|x'}}{{A_{a|x}}} = \bra{a}U_x^\dagger D U_{x} \ket{a'} = \braket{a}{a'} - 2 \bra{a}U_x^\dagger \ket{j} \bra{j}U_x \ket{a'}\,.
\end{equation}
Taking its absolute value, and using $\lvert \bra{i}U_x\ket{j} \rvert = 1/\sqrt{3}$, we obtain
\begin{equation}\label{eq:scalar_product_five_elements_set}
    \forall \,x'>x\,,x\neq 1\,, \quad |\braket*{A_{a|x}}{A_{a'|x'}}| =\delta_{aa'} \frac{1}{3} + (1-\delta_{aa'}) \frac{2}{3} \,.
\end{equation}
A similar computation for $D = 2\ketbra{j} - \id $ yields the same overlap as above. Hence, the bound~\eqref{appeq:overlap_qutrit} becomes
\begin{equation}\label{appeq:overlap_qutrit2}
    \forall  x'>x\,,\quad |\braket*{ \psi^B_{a|x}}{ \psi^B_{a'|x'}}|\leq
        \begin{cases}
        \sqrt{\lambda_{max}}  &\text{if}\quad x=1\,, \\ 
               1 - 2\lambda_{min}  &\text{if}\quad  x \neq 1\,, a=a'\,,\\
               1 - \lambda_{min}  &\text{if}\quad  x \neq 1\,, a\neq a'\,.
    \end{cases}
\end{equation}
Using Theorem~\ref{thmJohnston}, we can now provide a sufficient criterion for full nonlocality in dimension $d=3$.
\begin{result*}
    A bipartite qutrit state $\ket{ \psi}_{AB}=\sum_{i=0}^2\sqrt{\lambda_{i}}\ket{ii}$ is fully nonlocal if
    \begin{equation}\label{eq:qutrits_sufficient_condition_fnl}
     0.328 \approx \frac{22-\sqrt{259}}{18} \leq \lambda_{min} \leq \frac{1}{3}\, . 
    \end{equation}
\end{result*}
\begin{proof}
    The proof mimicks the one in App.~\ref{app:main_thm}. Consider the Gram matrix $G$ of the set of post-measurement states $\rho^B_{\boldsymbol{\alpha}} \coloneqq \{ \ket*{\psi^B_{a=\alpha_x|x}}\}_{x\in[5]}$ for any arbitrary $\boldsymbol{\alpha} = [\alpha_x]_{x\in[5]}$. The square of its Frobenius norm $(\norm{G}_F)^2$ is given by 
\begin{equation}
    \begin{aligned}
        (\norm{G}_F)^2 
        &= 5 + 2\sum_{x'=2}^5 |\braket*{ \psi^B_{a=\alpha_1|1}}{ \psi^B_{a=\alpha_{x'}|x'}}|^2 + 2 \sum_{x=2}^5 \sum_{x'>x} |\braket*{ \psi^B_{a=\alpha_x|x}}{ \psi^B_{a=\alpha_{x'}|x'}}|^2\,.
    \end{aligned}
\end{equation}
We now use Eq.~\eqref{appeq:overlap_qutrit2} to upper bound the second and third term above. For the third term, notice that $(1 - 2\lambda_{min}) < (1 -\lambda_{min})$, and hence the worst case is $\alpha_x \neq \alpha_{x'}$. However, since there are only $3$ possible outcomes for Alice, i.e., $\alpha_x,\alpha_{x'} \in \{0,1,2\}$, there must be at least one pair $(x,x')$ in the summation such that $\alpha_x = \alpha_{x'}$. Therefore,
\begin{equation}
        (\norm{G}_F)^2\leq 5 + 8 \lambda_{max} + 2 \cdot (1-2\lambda_{min})^2+2\cdot 5 \cdot (1-\lambda_{min})^2=18\lambda_{min}^2-28\lambda_{min}+8\lambda_{max}+17.
\end{equation}
    Using $\lambda_{max}\leq 1-2\lambda_{min}$, and imposing the condition $(\norm{G}_F)^2\leq 25/2$ from Theorem~\ref{thmJohnston}, we get the sufficient condition $ 36\lambda_{min}^2-88\lambda_{min}+25\leq 0$ for full nonlocality, which has the solution given in Eq.~\eqref{eq:qutrits_sufficient_condition_fnl}.
\end{proof}
\section{Proof of Theorem~\ref{thm:local_content}}\label{app:localcontent}
Here we give a proof of Theorem~\ref{thm:local_content} which provides a fundamental tool to obtain a lower bound on the local content of entangled states.
\begin{theorem*}
    Consider a $d$-dimensional entangled state $\ket{ \psi}=\sum_{i=0}^{d-1}\sqrt{\lambda_i}\ket{ii}$, where $\lambda_0\geq\lambda_1\geq\cdots\geq\lambda_{d-1}$. If Alice and Bob perform rank-$1$ measurements with at most $m_A$ and $m_B$ outcomes respectively on this state, then, for every measurement choice of Alice and Bob, there exists a POVM element $A = \ketbra{A}$ for Alice and $B=\ketbra{B}$ for Bob such that
    \begin{equation}\label{appeq:local_content_bound}
     \lvert \braket{ \psi}{A\otimes B} \rvert \ge \frac{\sqrt{\lambda_0}-\sqrt{\lambda_1}\sqrt{(m_A-1)(m_B-1)}}{\sqrt{m_A m_B}}\,.
    \end{equation}
\end{theorem*}
\begin{proof}
Let us begin by fixing explicitly the POVM elements $A = \ketbra{A}$ and $B=\ketbra{B}$ used in the statement of the theorem. For any $m_A-$outcome, rank-$1$ POVM $\{A_1,\dots,A_{m_A}\}$ of Alice, we must have
\begin{equation}
       \bra{0}\left(\sum_{l=1}^{m_A} \ketbra{A_l}\right)\ket{0} = \bra{0} \id \ket{0} \implies \sum_{l=1}^{m_A} \lvert\braket{0}{A_l}\rvert^2 = 1\,.
\end{equation}
Hence, for some POVM element $A \in \{A_1,\dots,A_{m_A}\}$, $\lvert\braket{0}{A}\rvert^2 \ge 1/m_A$. Similarly, for some POVM element $B \in \{B_1,\dots,B_{m_B}\}$ of Bob, $\lvert\braket{0}{B}\rvert^2 \ge 1/m_B$. Let us denote by $a_j$ the $j^{\text{th}}$ component of $\ket{A}$ in the Schmidt basis of $\ket{ \psi}$, i.e., $a_j=\braket{j}{A}$, and by $b_j$ the $j^{\text{th}}$ component of $\ket{B}$, i.e., $b_j=\braket{j}{B}$. Hence,
\begin{equation}\label{eq:local_content_outcome_bound}
    |a_0|^2 \ge \frac{1}{m_A} \quad \text{and} \quad |b_0|^2 \ge \frac{1}{m_B}\,.
\end{equation}
We now consider the scalar product
\begin{equation}
    \begin{aligned}
        \braket*{ \psi}{A\otimes B}&= \sum_{j=0}^{d-1}\sqrt{\lambda_j}\,\braket{j}{A} \braket{j}{B} = 
        \sqrt{\lambda_0}\, a_0b_0+\sum_{j=1}^{d-1} \sqrt{\lambda_j}\, a_j b_j \,,
    \end{aligned}
\end{equation}
and take its absolute value, which satisfies 
\begin{equation}\label{eq:local_content_ineq}
    \lvert \braket{ \psi}{A\otimes B}\rvert \ge \sqrt{\lambda_0}\left|a_0b_0\right|-\bigg|\sum_{j=1}^{d-1} \sqrt{\lambda_j} \, a_j b_j \bigg|\,.
\end{equation}
Furthermore,
\begin{equation}\label{eq:local_content_ineq2}
 \bigg|\sum_{j=1}^{d-1} \sqrt{\lambda_j}\, a_j b_j \bigg| \le \sum_{j=1}^{d-1} \sqrt{\lambda_j} \,\lvert a_j b_j\rvert  \le \sqrt{\lambda_1}\, \sum_{j=1}^{d-1} \lvert a_j b_j\rvert \,,
\end{equation}
where in the final step we have used the fact that $\lambda_1 = \max\{\lambda_1,\dots,\lambda_{d-1}\}$. Defining $\ket{u} \coloneqq (|a_1|,\dots,|a_{d-1}|)$, $\ket{v}\coloneqq (|b_1|,\dots,|b_{d-1}|)$ and applying the Cauchy--Schwarz inequality $\lvert \braket{u}{v} \rvert \le \sqrt{\braket{u} \braket{v}}$ we get,
\begin{align}\label{eq:local_content_ineq3}
\sqrt{\lambda_1} \,\sum_{j=1}^{d-1} \,\lvert a_j b_j\rvert\leq \sqrt{\lambda_1} \left(\sum_{i=1}^{d-1} |a_i|^2\right)^{1/2} \left(\sum_{j=1}^{d-1} |b_j|^2\right)^{1/2} .
\end{align}
Inequalities~\eqref{eq:local_content_ineq2},\eqref{eq:local_content_ineq3} imply that we can relax the inequality~\eqref{eq:local_content_ineq} as
\begin{equation}\label{eq:local_content_ineq4}
    \lvert \braket{ \psi}{A\otimes B}\rvert \ge \sqrt{\lambda_0}\left|a_0b_0\right|-\sqrt{\lambda_1} \left(\sum_{i=1}^{d-1} |a_i|^2\right)^{1/2} \left(\sum_{j=1}^{d-1} |b_j|^2\right)^{1/2}.
\end{equation}
Moreover, using Eq.~\eqref{eq:local_content_outcome_bound} we obtain
\begin{equation}\label{appeq:local_content_trace_bound_A}
    \sum_{j=1}^{d-1} |a_j|^2 = \left(\sum_{j=0}^{d-1} |a_j|^2\right) - |a_0|^2 = \braket{A} - |a_0|^2 \le \braket{A}-\frac{1}{m_A} \le \frac{m_A-1}{m_A} \,,
\end{equation}
where in the final step we have used $\braket{A} \le 1$ since the trace of a rank-$1$ POVM element is always smaller than $1$.
Similarly, on Bob's side we obtain
\begin{equation}\label{appeq:local_content_trace_bound_B}
    \sum_{j=1}^{d-1} |b_j|^2 \leq \frac{m_B-1}{m_B}\, . 
\end{equation}
Substituting the RHS of inequalities~\eqref{appeq:local_content_trace_bound_A},\eqref{appeq:local_content_trace_bound_B} in~\eqref{eq:local_content_ineq4}, and using $|a_0 b_0| \ge 1/\sqrt{m_A m_B}$, we obtain the desired result~\eqref{appeq:local_content_bound}.

\end{proof}

\end{document}